\documentclass[12pt]{article}

\usepackage{epsfig}

\usepackage{amssymb}
\usepackage{amsfonts}

\usepackage{color}
 
\def\be{\begin{equation}}
\def\ee{\end{equation}}
\def\ba{\begin{array}{c}}
\def\ea{\end{array}}

\newcommand{\bea}{\begin{eqnarray}}
\newcommand{\eea}{\end{eqnarray}}

\newcommand{\kt}{\rangle}
\newcommand{\br}{\langle}

\newtheorem{thm}{Theorem}

\newtheorem{lemma}[thm]{Lemma}

\newenvironment{proof}{\noindent
 {\bf Proof.}}{\hfill$\square$\vspace{3mm}\endtrivlist}

\begin{document}

\begin{center}

.

{\Large \bf

Feasibility and method of multi-step Hermitization
of crypto-Hermitian quantum Hamiltonians

}

\vspace{10mm}

\textbf{Miloslav Znojil}

\vspace{0.2cm}

\vspace{0.2cm}

The Czech Academy of Sciences, Nuclear Physics Institute,

 Hlavn\'{\i} 130,
250 68 \v{R}e\v{z}, Czech Republic

\vspace{0.2cm}

 and

\vspace{0.2cm}

Department of Physics, Faculty of Science, University of Hradec
Kr\'{a}lov\'{e},

Rokitansk\'{e}ho 62, 50003 Hradec Kr\'{a}lov\'{e},
 Czech Republic

\vspace{0.2cm}

 and

\vspace{0.2cm}

Institute of System Science, Durban
University of Technology, Durban, South
Africa

\vspace{0.2cm}

{e-mail: znojil@ujf.cas.cz}


\end{center}

\newpage

\section*{Abstract}

Even the most economical conventional
formulation of quantum mechanics
called Schr\"{o}dinger picture (SP)
need not be always computationally or conceptually
optimal. A Dyson-inspired
remedy can sometimes be sought in
a transfer of the calculations
from the conventional and correct SP Hilbert space ${\cal H}$
into a ``false but friendlier'' Hilbert
space ${\cal R}$.
The difference (i.e., non-equivalence) between ${\cal R}$
and ${\cal H}$ is given by a
judicious simplification of the inner product.
Although
the Hamiltonian $H$ itself appears non-Hermitian in ${\cal R}$,
its necessary
Hermitization
(mediated by a return to ${\cal H}$)
is usually feasible.
If not, another eligible remedy is described in the present paper.
Our basic idea lies in the introduction of
a ``chained'' auxiliary manipulation multiplet of inner-product spaces
${\cal R}_1$, ${\cal R}_2$, \ldots, ${\cal R}_N$.
At arbitrary $N$, the structure and
properties of the resulting generalized SP (GSP)
formulation of quantum mechanics of unitary systems
(characterized by ``hiddenly Hermitian'' observables)
are described and discussed in detail.
Nontrivial nature of the related
multi-step Dyson-map
composition laws and
correspondences to the
conventional Hermitian models are also revealed and
clarified.


\subsection*{Keywords}

.

quantum dynamics
in Schr\"{o}dinger picture;

sequences of auxiliary inner-product spaces;

apparently non-Hermitian generators of
unitary evolution;

factorization of the physical Hilbert-space metrics;

\newpage

\section{Introduction}

In conventional textbooks
the basics of quantum theory
of unitary systems
are usually explained
in Schr\"{o}dinger
representation {\it alias\,}
Schr\"{o}dinger picture  (SP, \cite{Messiah}).
The pure state
is represented there by a
time-dependent ket-vector element
of a Hilbert space
${\cal L}$. The evolution
is assumed generated by a self-adjoint Hamiltonian $\mathfrak{h}$.
The one-to-one correspondence between the
self-adjointness of $\mathfrak{h}$ and the unitarity of the
evolution was given the rigorous
mathematical form by Stone \cite{Stone}.
In this sense, the
Bender's and Boettcher's
claim \cite{BB}
that the
unitary evolution
could be
also generated by a non-Hermitian
Hamiltonian $H \neq H^\dagger$
sounded, initially, contradictory \cite{Streater}.

Fortunately,
the apparent paradox
found quickly its
origin in an
elementary terminological misunderstanding
(see, e.g., reviews \cite{Carl,SIGMA,ali}
or section \ref{qe} below).
It has been revealed that
the
Bender's and Boettcher's
Hamiltonians $H$ are
only non-Hermitian in a ``theoretically redundant'',
i.e.,
mathematically preferable but
manifestly unphysical auxiliary Hilbert space
${\cal H}_{math}$.
In this sense, these operators have to be
Hermitized in a way explained,
by Scholtz et al,
in review \cite{Geyer}. In their words,
given a suitable non-Hermitian Hamiltonian $H$, one can still work in
an innovated version of quantum theory which
``allows for the normal quantum-mechanical interpretation``.

In such a setting the unitary evolution is found
generated by a non-Hermitian Hamiltonian $H$ because
such an operator (called ``quasi-Hermitian'' \cite{Geyer}
{\it alias\,} ``crypto-Hermitian'' \cite{Smilga})
appears {\em Hermitizable}.
In applications the Hermitization
has to be visualized as a
constructive
transfer of representation of the quantum bound
states of interest from the auxiliary Hilbert space ${\cal H}_{math}$
to the traditional physical Hilbert space ${\cal L}$ of textbooks
or, better \cite{ali},
to one of its other suitable representations (say, ${\cal H}_{phys}$).

The phenomenological relevance as well as the
feasibility of the
procedure of Hermitization
was illustrated, using an exactly solvable model,
by Buslaev and Grecchi \cite{BG}. Rigorously, these authors
managed to show that the
unstable anharmonic-oscillator toy-model Hamiltonian
 \be
  H^{(BG)}=H_{}({\rm i}g,j)=-\frac{d^2}{dx^2}+
 \frac{j^2-1}{4\,r^2_{}(x)}+
 \frac{r^2_{}(x)}{4}-
 \frac{g^2\,r^4_{}(x)}{4}\,,
 \ \ \
 r_{}(x)=x-{\rm i}{\eta}\,,\ \ \ \eta>0\,
 \label{ope}
 \ee
defined as acting and manifestly non-Hermitian
in an unphysical Hilbert space ${\cal H}_{math}=L^2(\mathbb{R})$
can be assigned, at any real $j$ and positive
$g>0$, the entirely conventional
self-adjoint
{\em isospectral\,} partner
 \be
 \mathfrak{h}^{(BG)}=Q(g,j)=-\frac{d^2}{dx^2}
 -(gx-1/2)\,j +
 (gx-1)^2\,x^2\,,\ \ \ \ \  x \in \mathbb{R}
 \,
 \label{upe}
 \ee
defined as acting
in another, physical  Hilbert space ${\cal L}=L^2(\mathbb{R})$.

Remarkably enough, even such a
most elementary sample of Hermitization
$H^{(BG)} \to \mathfrak{h}^{(BG)}$
appeared to
require the introduction of several
intermediate auxiliary isospectral operators
(their list
contains cca six items,
the explicit form of which
may be found in {\it loc. cit.}).
On this ground
one may suspect that
the {\it feasible\,}
process of
Hermitization
should also
have a
{\em multi-step}, $N-$step
structure in general.
This is the idea which also served as
the main motivation of our present
study.

In the first stage of development
(cf. section \ref{qe})
we will
keep $N=1$ and
simplify and reformulate
the very concept of the Hermitization.
For this purpose
we will just review the two
complementary existing forms of the theory.
In its older version
(cf. review \cite{Geyer} or subsection \ref{DT} below) one
manages to circumvent
the (in general, difficult) description of the
direct change $H\to \mathfrak{h}$ of the Hamiltonian.
This is mediated by
introduction of another, third  Hilbert
space ${\cal H}_{phys}$ \cite{SIGMA}.
In subsection \ref{ufr} we will briefly review
the newer and better known
specific approach to the technicalities
known as ${\cal PT}-$symmetric
quantum mechanics \cite{Carl}.
The practical feasibility of the
Hermitization process is
enhanced there simply by a judicious
factorization of the
inner-product metric
(cf. formula (\ref{fak}) below).
Naturally,
the ultimate amended physical space ${\cal H}_{phys}$ must be
constructed as equivalent to its ``user-unfriendly''
textbook predecessor ${\cal L}$.
By construction \cite{Carl,ali}, Hilbert space ${\cal H}_{phys}$
only differs from manifestly unphysical ${\cal H}_{math}$
by an amendment and factorization of the inner-product metric
(i.e., in our present notation, of operator $\Theta$)
in a way illustrated by the
diagram of Eq.~(\ref{humdi}) below.

In section \ref{heha} we will emphasize that both of the latter amendments
are essential. They will be shown to lead to
a decisive ultimate simplification of
correspondence between spaces ${\cal L}$ and ${\cal H}_{phys}$.
In the same section we will also
recall our recent letter~\cite{PLA}
where we described the first nontrivial implementation
of the idea.
Basically, we just considered $N=2$ and
proposed there a modified and slightly amended
version of the Bender-inspired ${\cal PT}-$symmetric
quantum mechanics. Although we encountered there
certain conceptual obstacles (so that we simply
did not think about $N>2$),
still, letter \cite{PLA} can be perceived
as a basic inspiration of our present
message.

Briefly, our present proposal can be characterized as
an implementation of the
same non-Hermitian model-building
strategy using arbitrary $N \geq 2$.
The explicit form of the innovation will be
outlined
in section~\ref{multikulti}, while its
physical
background and consequences
will be discussed in section \ref{opere}.
In a complementary section \ref{dymas}
we will finally add a few comments on the
correspondence between the present
$N-$step
Hermitizations of the Hilbert-space-changing form
${\cal H}_{math}\to {\cal H}_{phys}$
and the
older, Dyson-inspired \cite{Geyer,Dyson}
and Hamiltonian-changing
model-building flowcharts.

The  summary and outline of some
potential practical benefits
of our present generalized SP (GSP)
reformulation of quantum mechanics
will finally be given
in section~\ref{sumys}.

\section{Quantum mechanics using two inner products ($N=1$)\label{qe}}

Even in the purely methodical context
the Buslaev's and Grecchi's construction connecting the
asymptotically repulsive ``wrong sign''
complex potential in (\ref{ope})
with the
asymptotically well-behaved
double-well potential in (\ref{upe})
is exceptional. Indeed,
both of the related partner Hamiltonians
possessing the discrete and
real bound-state spectra
are still just ordinary
differential operators of the second order.
In contrast,
the analogous would-be isospectral partners $\mathfrak{h}^{(BB)}$
of all of the
Bender's and Boettcher's \cite{BB} models of the form
$H^{(BB)}=-d^2/dx^2+V^{(BB)}(x)$
have only been constructed approximatively \cite{Carl,ali,117}. Moreover,
even this indicated that
their form (which,
as required, was self-adjoint in ${\cal L}$)
was non-local and extremely complicated,
represented by
differential operators of infinite order.
After all, such an explicit construction
experience only reconfirmed the expectable
contrast
between a guaranteed simplicity (i.e., computational tractability) of
``input'' $H$ and
a prohibitive
complexity of a generic ``output'' $\mathfrak{h}$.

The same contrast can be also detected
in an opposite extreme of
practical realistic calculations.
In the above-cited review \cite{Geyer}, for example,
the authors' study of the Hermitization flowchart
 \be
 H\,\to\,\mathfrak{h}=
 \Omega\,H\,\Omega^{-1}\,
 \label{opprvak}
 \ee
was in fact inspired by
the Dyson's realistic many-body calculations \cite{Dyson}.
In them, the arrow in Eq.~(\ref{prvak}) had only been
drawn in opposite direction,
 \be
 \mathfrak{h}^{(Dyson)}\,\to\,
 H^{(Dyson)}= \Omega^{-1}\,\mathfrak{h}^{(Dyson)}\,\Omega\,.
 \label{prvak}
 \ee
Still, the motivation of the
use of the isospectral
mapping $\Omega$
remained precisely the same:
There was no doubt about the choice between
the realistic but user-unfriendly many-fermion
Hamiltonian $\mathfrak{h}^{(Dyson)}$
and its judiciously pre-conditioned
isospectral bosonic partner $H^{(Dyson)}$ of Eq.~(\ref{prvak}).
Indeed, as long as the Dyson's choice of
the preconditioning operator
(or, in fact, matrix) $\Omega^{(Dyson)}$
was supported by intuition, the mapping (\ref{prvak})
could really lead to a simplification
of Schr\"{o}dinger equation,
especially when Dyson decided to
work with the most general non-unitary mappings such that
 \be
 \Omega^\dagger \Omega = \Theta \neq I\,.
 \label{amenes}
 \ee

\subsection{Dyson's transformation and crypto-Hermitian Hamiltonians\label{DT}}

During the application of the conventional textbook SP approach to
certain many-particle quantum systems the convergence of variational
methods may become prohibitively slow. In these cases the
Hamiltonian is, typically, a complicated partial differential
operator $\mathfrak{h}$ acting in a complicated many-body Hilbert
space ${\cal L}=L^2(\mathbb{R}^d)$ with a large
dimension-representing exponent $d$. The convergence becomes
particularly slow when the particles are fermions. In this case,
indeed, the system has to obey the Pauli exclusion principle
\cite{Jenssen} so that ``it has become customary to map Hermitian
fermion operators onto non-Hermitian boson operators'' \cite{Geyer}.
Purely empirically it has been revealed, originally by Dyson
\cite{Dyson}, that the calculations (typically, of the low-lying
spectra) may perceivably be accelerated via a judicious {\it a
priori\,} simulation of the correlations. In the language of
mathematics this means that it makes sense to replace the initial
``fermionic'' Hilbert space ${\cal L}$ by an auxiliary (and,
formally, non-equivalent) ``bosonic'' Hilbert space, say, ${\cal
H}_{math}$. In the latter (and, by assumption,
user-friendlier) space the system becomes described by a new
Hamiltonian (say, $H$). Naturally, a successful (often, just
intuition-supported) estimate and simulation of the structure of the
correlations of fermions has often been found to imply a decisive
acceleration of the convergence of the variational treatment of the
bosonic operator $H$. In the language of numerical mathematics one
may speak about an isospectral preconditioning (\ref{prvak})
of the Hamiltonian.

The difference between the ``good'' and ``bad'' choice of the
operator $\Omega$ mapping ${\cal H}_{math}$ onto ${\cal
L}$ is strongly model-dependent. One of the conditions of success is
that this mapping (often called Dyson mapping) has to be as general
and adaptable as possible, i.e., in particular, non-unitary
(cf. (\ref{amenes})).
In a way explained in \cite{SIGMA} the new, preconditioned,
user-friendlier Hamiltonian $H$ can be perceived as defined in
another Hilbert space ${\cal H}_{math}$ which is, due to property
(\ref{amenes}), not equivalent to its physical predecessor ${\cal
L}$: In the Dyson's realistic calculations \cite{Dyson}, for
example, the fermionic states in ${\cal L}$ were represented by the
bosonic states in ${\cal H}_{math}$.

In
general, the Dyson's rules (\ref{prvak}) and (\ref{amenes}) render $H$
non-Hermitian
in ${\cal H}_{math}$,
 \be
 H \neq H^\dagger=\Theta\,H\,\Theta^{-1}\,.
 \label{quasi}
 \ee
For any preselected candidate $H$ for Hamiltonian the
constraint~(\ref{quasi}) can be read as specifying an amended,
alternative, correct and physical inner product in Hilbert space
${\cal H}_{phys}$ \cite{ali}. For this reason, the well known Stone
theorem~\cite{Stone} is not violated because the unitary evolution
of the system is in fact generated, by non-Hermitian $H$, in the
auxiliary Hilbert space ${\cal H}_{math}$. This space is not
equivalent to its physical partner ${\cal H}_{phys}$. Thus, one
should rather call $H$ crypto-Hermitian~\cite{Smilga}.

In this setting it is probably useful to point out that the
preservation or failure of the ``duality'' or ``complementarity''
between  ${\cal H}_{math}$ and  ${\cal H}_{phys}$ may be fairly
sensitive to the detailed properties of the Hamiltonian in question.
For a word of warning one does not even need to go to the
non-Hermitian quantum field theory where the Stone theorem need not
apply. Indeed, it is fully sufficient to see that multiple subtle
formal problems like, e.g., the possibility of the non-existence of
the metric may emerge in a finite-dimensional model (see, e.g.,
section Nr. 3 in review \cite{Geyer}).

The main aim of transformation (\ref{prvak})
is that the description of dynamics
(originally provided by $\mathfrak{h}$)
is now shared by the two operators (viz., by $H$ and $\Omega$ or $\Theta$).
In practical calculations, this makes the Dyson-inspired SP more flexible,
in principle at least.
Still, all of the mathematical operations
have to be performed,
exclusively, in the auxiliary
Hilbert space
${\cal H}_{math}$.
Although the latter space seems to
play just a technical role,
its key merit is that
it can be re-read,
after the mere {\it ad hoc\,}
redefinition of the inner product
(see \cite{ali,Geyer} or formula (\ref{reccu}) below, with $K=2$)
as a new, unitarily equivalent representation
${\cal H}_{phys}$ of
the ``old and missing''  Hilbert space
${\cal L}$.
The conventional probabilistic interpretation
of the unitary quantum dynamics gets restored.

The non-Hermiticity of $H$
in the working space ${\cal H}_{math}$
leads to
the necessity of using some
less straightforward methods of solution
of the underlying Schr\"{o}dinger equation.
In many models,
the merits of the adaptability of the
non-unitary
mapping
(often called, in the light of papers \cite{Dyson},
Dyson map) proved to
prevail
(cf., e.g., the success and productivity of this approach in
nuclear physics \cite{Jenssen}).
Still, one of the unavoidable consequences
of the non-unitarity of $\Omega$
is that
the conventional interpretation of the
eigenstates of $H$
has to be
modified.
In the manner described in reviews \cite{SIGMA,Geyer}
the situation may be clarified and
visualized using the following diagram
 \be
 \label{humdi}
  \ba
   \ \  \begin{array}{|c|}
 \hline
  \ {\rm inaccessible \ physical\ Hilbert\ space} \
 {\cal L}\ {\rm of\ textbooks} \\
 \hline
 \ea
\\ \ \ \ \ \
\ \ {\rm (Dyson's\ map}\
 \Omega) \   \nearrow \ \  \
 \ \ \ \ \ \
 \ \ \ \ \ \
\ \ \ \ \  \nwarrow  \searrow  \  { \rm (equivalence)}\ \ \  \\
\ \
 \begin{array}{|c|}
 \hline
  \ {\rm friendlier\ representation  \ space} \
 {\cal H}_{math}\ \\
  \hline
 \ea\
 \stackrel{{\rm (metric}\ \Theta)}{ \longrightarrow }
 \
 \begin{array}{|c|}
 \hline
  \ {\rm alternative \ physical\  space} \
 {\cal H}_{phys}\ \\
 \hline
 \ea
\\
\\
\ea
 \ee
This diagram indicates that
in certain specific non-Hermitian-Hamiltonian models
one has to separate
the probabilistic
physical predictions formulated in ${\cal H}_{phys}$
from the explicit calculations
performed, much more efficiently, in
${\cal H}_{math}$.
The
space ${\cal L}$ and the SP
Hamiltonian $\mathfrak{h}$
have to be abandoned
and replaced by
the doublet ${\cal H}_{math/phys}$
and by the preconditioned
Hamiltonian $H$, respectively.
A very general crypto-Hermitian SP (CHSP) of review \cite{Geyer}
is born.

The essence of the CHSP formalism
has most concisely been explained by Mostafazadeh \cite{ali}.
He put emphasis upon the interplay
of the inner-product structures in ${\cal H}_{math/phys}$.
He emphasized that the
standard physical
role is played by
${\cal H}_{phys}$.
Although the merits
of
the introduction of another,
apparently redundant Hilbert space
${\cal H}_{math}$
seem less obvious, its introduction
separated the interpretations
(in ${\cal H}_{phys}$)
from calculations
(in ${\cal H}_{math}$) and facilitated the
qualitative as well as quantitative predictions
(cf., once more, the above-mentioned
Dyson's study of ferromagnetism \cite{Dyson}).
At the same time,
during
the implementation of the general CHSP theory
people encountered also serious technical obstacles
(some of them were
listed on p. $1216\ $ of review
\cite{ali}).
For this reason,
only
the most recent simplifications made the theory
really widely known and
successful \cite{Carlbook}.

\subsection{${\cal PT}-$symmetric quantum mechanics\label{ufr}}

The CHSP approach can be perceived as
equivalent to
the standard quantum mechanics, being distinguished just by the
conversion of the conventional SP Hamiltonian $\mathfrak{h}\,$ into its
less standard
isospectral representation $H$
defined as acting
in ${\cal H}_{math}$ and
amenable to necessary Hermitization.
Nevertheless, the
CHSP
formalism
has only recently been converted,
via decisive simplifications,
into
one of the most influential and popular
versions of SP,
widely known as
${\cal PT}-$symmetric quantum mechanics (PTQM, \cite{Carl}).

In a slightly less general PTQM
setting, the key role is
still played by the
details of the metric-mediated transition from the auxiliary
Hilbert space
${\cal H}_{math}$ to its correct physical partner
${\cal H}_{phys}$. In a way indicated by the horizontal arrow
in diagram~(\ref{humdi}) the fundamental
message that
the Hamiltonian $H$ must be Hermitian in ${\cal H}_{phys}$
remains unchanged.
In the literature, unfortunately, this message is often obscured by
the diversity of notation conventions
(see their sample in Table Nr. 1 of \cite{NIP}).
We will use here, therefore,
the properly modified Dirac's
bra-ket formalism \cite{Messiah}
and the maximally
compact abbreviations
for ${\cal H}_{phys}= {\cal R}_{0}$
and
${\cal H}_{math}= {\cal R}_{N_{}}$.
This will enable us to distinguish easily
between the quantum mechanics of textbooks
(in which we may put $N_{}=0$ to imply that
${\cal H}_{math}\ \equiv \ {\cal H}_{phys}\ \equiv \ {\cal L}_{}$)
and the menu of the two-space CHSP scenarios
in which one may choose $N_{}\geq 1$
emphasizing
the non-equivalence between
${\cal H}_{math}= {\cal R}_{N_{}}$
and ${\cal H}_{phys}= {\cal R}_{0}$.

One of the sources of the
user-friendliness of the general CHSP formalism as
described by Scholtz et al \cite{Geyer}
is that once we put, for simplicity, $N_{}= 1$,
we just have to consider
the doublet
 \be
 \{{\cal H}_{math},{\cal H}_{phys}\}\ = \
 \{{\cal R}_{1},{\cal R}_{0}\}
 \,
 \label{doupl}
 \ee
of relevant
Hilbert spaces.
This clearly differs from the
conventional quantum theory in which
one sets $N_{}= 0$.
After innovation (\ref{doupl})
the Hilbert space ${\cal H}_{math}$ becomes unphysical
and the Hamiltonian $H$ becomes manifestly non-Hermitian
in this space.

During the necessary
Hermitization of $H$
the main technical task is
the
construction
of the correct physical metric operator
$\Theta$ which would be compatible with
the Hamiltonian
crypto-Hermiticity
{\it alias\,} pseudo-Hermiticity \cite{ali}
{\it alias\,} quasi-Hermiticity \cite{Geyer}
condition~(\ref{quasi}).
It is worth a comment that the
key role of
operator $\Theta$
(with the upper-case Greek-letter symbol proposed in \cite{SIGMA})
is in a sharp contrast
with a lack of
its sufficiently widely accepted denotation. In
Refs.~\cite{Dyson}, \cite{Geyer} or \cite{ali}, for example,
these
physical Hilbert-space metrics may be found denoted
by the very different symbols like $F$, $T$ or $\eta_+$,
respectively.

In the most straightforward Bender-inspired
PTQM approach \cite{Carl} the
metric operator is constructed
in the form of product
 \be
 \Theta_{(Bender)}={\cal PC}\,
 \label{fak}
 \ee
where ${\cal P}$ is parity
and where
the symbol ${\cal C}$ denotes the so called charge \cite{BBJ}.
The introduction of such an ansatz
proved well motivated. In many models
it appeared to offer
a perceivable simplification of the
calculations as well as
of the subsequent necessary extraction of the
model-dependent
testable and also, in principle,
falsifiable predictions \cite{ali,Carlbook}.

\section{Quantum mechanics using three inner products ($N=2$)\label{heha}}


Our present paper will be devoted to
a certain conceptual and methodical extension and completion
of the CHSP and PTQM approaches.
We felt motivated by the observation that
in a broader area of physics
both of these recipes
appeared comparatively
difficult to implement.
Various authors offered
remedies involving the
restriction of attention to the bounded-operator
Hamiltonians \cite{Geyer} or the use of
specific models in which
the SP Hamiltonian $H$ is
${\cal PT}-$symmetric (see \cite{Carl}).
It is worth adding that
the latter approach was
based on a
parity-related (i.e., ${\cal P}-$related)
and time-reversal-related (${\cal T}-$related) assumption
$H{\cal PT}={\cal PT}H$
known as
${\cal PT}-$symmetry
of
the Hamiltonian.
This was an intuitively appealing feature
which made the PTQM approach particularly
popular, often even beyond its original scope and
restriction
to the closed and unitary quantum systems~\cite{Christodoulides}.

\subsection{Intermediate space}

Strictly speaking, several differences between
the closely related CHSP and PTQM unitary models
are non-trivial \cite{ali}. This fact happened to be obscured
by the slightly misleading current terminology.
{\it Pars pro toto}, let us mention that the
Hamiltonian operators $H$ which are
all required to possess the real spectrum and which are all
required self-adjoint in the correct physical Hilbert space
${\cal H}_{phys}$
are still called non-Hermitian in the literature.
Presumably, the reasons are psychological:
All of the
necessary two-space calculations are, naturally,
performed in just one of the frames, viz., in
the auxiliary and manifestly unphysical Hilbert space
${\cal H}_{math}$
in which the Hamiltonians
really {\em are\,} non-Hermitian.

In our recent letter \cite{PLA} we
paid more attention to the terminology. We
decided to
choose $N_{}= 2$
and to
replace the CHSP
Hilbert-space doublet (\ref{doupl})
by triplet
 \be
 \{{\cal H}_{math},{\cal H}_{intermediate},{\cal H}_{phys}\}\ = \
 \{{\cal R}_{2},{\cal R}_{1},{\cal R}_{0}\}
 \,.
 \label{tripl}
 \ee
The introduction of the intermediate inner-product space
helped us to
throw new light
on the terminology as well as
on
the fundamental operator-product (\ref{fak}).
In the resulting intermediate-space SP (ISP)
version of the
formalism
the correct probabilistic physical interpretation
of a given set of some
preselected candidates
$\Lambda_m$ for the operators of observables
with $m=0,1,\ldots,{M}\,$
has been clarified.

On these grounds,
the Hermitization
did not proceed directly
from ${\cal H}_{math}$ to ${\cal H}_{phys}$
(as, for example, in Ref.~\cite{Geyer}) but rather
indirectly,
via the third inner-product
space (i.e., Hilbert or Krein space) ${\cal H}_{intermediate}$.
A key to the comparison
of ISP and PTQM
has been found in the charge.
Indeed, in
the PTQM framework
the theory only becomes consistent
after one guarantees the ${\cal PCT}-$symmetry
$\,H{\cal PCT}={\cal PCT}H\,$
of the Hamiltonian \cite{BBJ}.
In parallel, in the ISP context
the charge appeared to play
a double role, i.e., not only the original role of
component of the correct physical inner-product metric
(\ref{fak}) which determines
the geometry in ${\cal H}_{phys}={\cal R}_0$, but also
a new role of an auxiliary metric operator
in
${\cal H}_{intermediate}={\cal R}_1$.

For the sake of
an enhancement of
clarity we will now change the notation and abbreviate
${\cal C}=Z_1$
(emphasizing the
geometry-determining role of the charge in ${\cal R}_1$) and
${\cal P}=Z_2$
(underlining the analogous
auxiliary-metric-operator role of
the -- possibly, generalized -- parity
in ${\cal H}_{math}={\cal R}_2$).
Another
abbreviation
will be used to emphasize the privileged status
of the specific physical metric
$\Theta_{(Bender)}={\cal PC}=Z_2Z_1=Y_2\,$ of Eq.~(\ref{fak}).
Marginally, we might add that
in the literature
the reference to the ``parity'' survived
even when the operator
${\cal P}=Z_2\,$ itself
(re-denoted by Greek lower-case $\eta$ in \cite{ali}
and required to be self-adjoint in ${\cal R}_2$,
$Z_2=Z_2^\dagger$)
ceased to be immediately connected
with spatial reflection (check also \cite{shendr}).

\subsection{Hermitian conjugation as an ambiguous,
inner-product-dependent concept}

Needless to emphasize,
the unmodified
Dirac's notation conventions can only be used
in the
mathematical-space extreme of a
maximal Hilbert-space subscript
$j=N_{}$.
Thus, at $N_{}=2$,
the conventional
Hermitian conjugation of an operator
can be written in the conventional Dirac's form
$ \Lambda \to \Lambda^{\dagger}$ in
${\cal H}_{math}={\cal R}_{2}$.
Otherwise,
in every other space ${\cal R}_{j}$
with $j<N_{}$
we will mark
the
generalized
Hermitian conjugation of operators as follows,
 $$
  \Lambda \to  \Lambda^{\ddagger(j)}\,
 $$
i.e.,
by
a dedicated space-dependent superscript.

Concerning the Hermitian conjugations of the ket vectors,
the basic inspiration of its present, modified-Dirac
denotation lies in a
rather elementary fact
that any Hilbert space
(with the
inner product $\br \psi_a|\psi_b\kt $ linear in the second, ket-vector argument)
can be
perceived as an ordered pair $[{\cal V},{\cal V}']$ of a linear
topological vector space ${\cal V}$ (of conventional ket-vector elements
$|\psi\kt$) and of the dual
vector space of linear functionals
marked by the prime,  ${\cal
V}'$ (see, e.g., p. 246 in textbook
\cite{Messiah} where, incidentally, also one of the
best detailed explanations of the standard
bra-ket notation conventions
can be found).

As long as the three
Hilbert spaces (\ref{tripl}) have to
share the same linear
ket-vector-space
component ${\cal V}$ (cf. Ref.~\cite{PLA}),
it is necessary to individualize, by the notation,
the respective duals ${\cal
V}'$, i.e., the respective invertible antilinear correspondences
(we will write ${\cal T}: {\cal V} \to {\cal V}'$).
It would be insufficient to use the same Dirac-recommended
bra-vector symbol $\br \chi|$
for the denotation of the elements of all
of the different dual spaces
of our interest.

An appropriately adapted notation will be used here, therefore.
In essence, we will first define the inner
product in ${\cal H}_{phys}={\cal R}_0$
via an explicit specification of the
correct physical mappings ${\cal
T}_{phys}$. Only {\em after\,}
such a preparatory step
one can formulate the phenomenological predictions
deduced from the evaluation of the
corresponding matrix elements
of the relevant observables.
For this purpose
we
will use the matrix-multiplication-resembling left-action
convention to define
 \be
 {\cal T}_{phys}: |\psi\kt \ \to \ \br \psi|\,{\Theta}\,,
 \ \ \ \
 {\Theta}={\Theta}^\dagger\,
 \label{wide}
 \ee
Here the Hilbert-space-metric operator $\Theta$
is precisely the one entering the
crypto-Hermiticity constraint~(\ref{quasi}).
In this context it is important to re-emphasize
that the choice of the metric $\Theta$
only becomes admissible when one guarantees that
this operator is
self-adjoint in ${\cal H}_{math}$ \cite{Geyer}.

In connection with Eq.~(\ref{wide})
we should note that multiple
nontrivial technical problems may
emerge when the dimension of our
topological vector space
${\cal V}$ of kets is infinite.
This aspect of the theory
must carefully be discussed
in the framework of functional analysis \cite{book}.
As a complementary further reading
in this direction let us mention
the old Dieudonn\'{e}'s paper \cite{Dieudonne},
or the recent Refs.~\cite{ATbook} and
\cite{fabio}.
This being said, it is useful to add, explicitly, that
in any case, the Hilbert-space metric
operator $\Theta$ must necessarily be chosen
invertible and positive definite and bounded,
with bounded inverse \cite{Geyer,book}.

One of the other and most important related technical obstacles
is that in the general CHSP framework
our choice of the ``physical'' inner product
(i.e., of the invertible antilinear Hermitian-conjugation
correspondence (\ref{wide}))
is, given the Hamiltonian, {\em non unique}.
An exhaustive discussion of this topic can be found in
the literature~\cite{Messiah,ali,Lotor}.
The ambiguity of the inner product is only irrelevant
in ${\cal H}_{math}={\cal R}_{N_{}}$
where we may use the conventional
Dirac notation
 \be
 {\cal T}_{math}: |\psi\kt \ \to \ \br \psi|
 \label{e2}
 \ee
and where we may define the Hermitian conjugation operation
as a mapping in which $|\psi\kt$
is a column vector while $\br \psi|$ is perceived
as its row-vector transposition
containing complex-conjugate elements.

Let us now fix
$N_{}=2$ and let
us consider
the triplet $\{{\cal R}_{j}\}$ of
the inner-product spaces (\ref{tripl}) as in \cite{PLA}. Then, the
mathematical extreme of Eq.~(\ref{e2}) might optionally be
marked by a subscripted
bracketed $j=N_{}$,
 \be
 {\cal T}_{[N_{}]}: |\psi\kt \ \to \  \br \psi|\
  \equiv \ \br_{[N_{}]} \psi |\,.
 \label{dva2}
 \ee
At the other two $j<N_{}$
the subscripted index becomes obligatory.
Thus, at the correct and physical $j=0$ extreme of Eq.~(\ref{wide})
with $\Theta=Y_2$
we will write
 \be
 {\cal T}_{[0]}: |\psi\kt \ \to \  \br_{[0]} \psi|\
 \equiv \ \br \psi |\,Y_2\,.
 \label{dva0}
 \ee
As long as the auxiliary metric $Z_2=Z_2^\dagger\,$
in ${\cal H}_{math}={\cal R}_2\,$
is tractable as the (possibly, generalized)
parity \cite{shendr}, we will finally define the
intermediate-space Hermitian conjugation
as follows,
 \be
 {\cal T}_{[1]}: |\psi\kt \ \to \  \br_{[1]} \psi|\
 \equiv \ \br \psi |\,Z_2\,.
 \label{dva1}
 \ee
The specific
PTQM (or rather, with $N_{}=2$, ISP)
realization of the general
CHSP framework
can be now characterized by
formula (\ref{fak}). It
suppresses the
ambiguity of the inner-product metric
$\Theta_{(Bender)}=Y_2=Z_2\,Z_1$.
It also implies that
besides the Hamiltonian $H$
(alternatively denoted here as $Z_0$), also
the
charge $Z_1$ is kept observable.

In the original PTQM proposal \cite{BBJ} the formula
for $\Theta_{(Bender)}$
was
required to contain a {\em self-adjoint\,}
operator of charge.
In the alternative
ISP framework of paper \cite{PLA}
the same factorization ansatz
admited a
more general charge.
In contrast to its status in PTQM, it was admitted
{\em non-Hermitian\,} in ${\cal H}_{math}$. In our
present terminology
the charge $Z_1$
acquired the physical meaning of
another
non-Hermitian and
crypto-Hermitian but
Hermitizable observable-representing operator.

\section{Quantum mechanics using $K-$plets of inner products\label{multikulti}}


The choice of $N_{}=2$
which characterizes the ISP Hermitization of Ref.~\cite{PLA}
can be generalized.
The number $N_{}$ of the auxiliary,
manifestly unphysical
inner-product spaces ${\cal R}_j$ with $j> 0\,$
can be,
in the resulting
GSP extension of the ISP formulation of quantum mechanics,
arbitrarily large.
The ISP triplet
(\ref{tripl}) of the inner-product spaces
sharing the same ket-vector elements of ${\cal V}$
will be replaced  by the $K-$plet
 \be
 \{{\cal H}_{math},{\cal R}_{K-2},{\cal R}_{K-3},\ldots,{\cal R}_{1},
 {\cal H}_{phys}\}\,
 \label{multpl}
 \ee
with
$K=N_{}+1$ and with alternative
symbols for
the extremely mathematical inner-product space ${\cal H}_{math}= {\cal R}_{K-1}$
and for the extremely physical Hilbert space ${\cal H}_{phys}= {\cal R}_0$.

\subsection{The inner-product-dependent Hermitian conjugations}

For any preselected operator $\Lambda$ defined as acting upon kets $|\psi\kt
\in {\cal V}$, its
Hermitian conjugate partners will be
inner-product-dependent. At any $K$
these partners
will form the $K-$plet
 \be
 \Lambda^{\ddagger(0)}\,,\,
 \Lambda^{\ddagger(1)}\,,\,
 \ldots\,,\,
 \Lambda^{\ddagger({K}-2)}\,,\,
 \Lambda^{\ddagger({K-1})}\,
 \label{sedm}
 \ee
with
elements marked
by the inner-product-dependent superscripts.
In every space ${\cal R}_{j}$
we will consider
an operator $Z_j\ $ which will
be required self-adjoint in the respective space,
 \be
 Z_j=Z_j^{\ddagger(j)}\,,\ \ \ \
 j=0,1,\ldots,{K-1}\,.
 \label{tri}
 \ee
This assumption will enable us to treat, in
the GSP framework, every $Z_j$ as an
analogue of the ISP metric $\Theta$ of
Eq.~(\ref{wide}).
This means that
every such an operator will have to satisfy the
consistency conditions as listed, e.g.,
in Eq. Nr. (2.1) of
the physics-oriented review \cite{Geyer}.
These mathematical
conditions involve the domain-completeness (cf. Nr. (2.1a)),
Hermiticity (2.1b) (equivalent to
our present
Eq.~(\ref{tri})),
positive definiteness (2.1c)
and the metric-boundedness constraint
Nr. (2.1d) plus, let us add, the metric-inverse boundedness constraint
missing in \cite{Geyer}.
Still, the authors of
review \cite{Geyer} formulated a consistent theory because they managed to
circumvent certain formal difficulties
as mentioned by Dieudonn\'{e} \cite{Dieudonne} by restricting their
attention, drastically, just to the not too realistic models
possessing bounded-operator Hamiltonians. In contrast, Dieudonn\'{e}
himself revealed that once the class of Hamiltonians
remains sufficiently general, the ISP metric $\Theta$ need not exist
at all.

In 2015 the situation has been reconsidered and
summarized in the mathematically
oriented monograph \cite{book} where the readers may find that for
unbounded operators, a sufficiently general and still mathematically
satisfactory formulation of sufficient conditions of the
applicability of the theory is still an open problem. At the same
time, there are no formal obstacles in all of the phenomenological
applications working with finite matrices. In our present paper we
decided to pay attention, predominantly, just to the algebraic
aspects and amendments of the theory. This means that our readers
will either keep in mind the sufficiently elementary (and,
typically, finite-dimensional) implementations of the theory or,
alternatively, they will recall the highly sophisticated specialized
literature ({\it pars pro toto}, let us recommend Ref.~\cite{ATbook}
for introductory reading).

Somewhere in between the two extremes, one may consider the
separable (i.e., still sufficiently general) Hilbert spaces ${\cal
R}_j=[{\cal V},{\cal V}']$ marked by a subscript
$j=0,1,\ldots,{K-1}$. We will assume that the set ${\cal V}$ of the
ket vectors remains the same (i.e., $j-$independent)) while the dual
sets ${\cal V}'$ of the linear functionals \cite{Messiah} become
$j-$dependent, ${\cal V}'={\cal V}'_{[j]}$. Every space ${\cal R}_j$
is self-dual. The correspondence between a ket $|\psi\kt \in {\cal
V}$ and its dual can be perceived as a result of action of an
invertible antilinear operator ${\cal T}_{[j]}$. The $K-$plet of
Hermitian conjugations yields the subscript-dependent bra-vectors,
 \be
 {\cal T}_{[j]}: |\psi\kt \ \to \  \br_{[j]} \psi|\,,
 \ \ \ \  j=0,1,\ldots, {K-1}\,.
 \label{dva}
 \ee
We will drop again the highest
subscript $j={K-1}$ as redundant,
$\br_{[{K-1}]} \psi| \,\equiv\, \br \psi|$. This
underlines
that the Hilbert
space ${\cal R}_{K-1}={\cal H}_{math}$ is a privileged
one, preferred in the calculations.

Next,
we may
generalize
the two ISP rules~(\ref{dva1}) and~(\ref{dva0})
and set
$\br_{[{K-2}]} \psi|=\br \psi|\,Z_{K-1}$
and
$\br_{[{K-3}]} \psi|\,Y_{K-1}$, respectively.
These
bra-vectors
are just the first two
transforms of
the same preselected
element. The list is easily completed at any $K> 3$.
Using an antilexicographically
ordered set of arbitrary left-acting operators we may
define
 \be
 \br_{[{K}-2]} \psi|=\br \psi|\,Z_{K-1}\,,\ \ \
 \br_{[{K}-3]} \psi|=\br \psi|\,Y_{K-1}\,,\ \ \
  \br_{[{K}-4]} \psi|=\br \psi|\,X_{K-1}\,,\ \ \ldots\,,
  \label{toe}
 \ee
with the last-item
$\br_{[0]} \psi|=\br \psi|\,\Theta$.

In GSP it makes sense
to replace
the $(K-1)-$subscripted operators entering Eq.~(\ref{toe})
by the sequence of the $Z_j\, $
operators
with
$ j={K-2}, {K-3},\ldots$.
Such a replacement can be realized via
the sequence of linear equations
 \be
 Y_j=Z_j\,Z_{j-1}\,,
 \ \ \ \
 X_{j}=Z_{j}\,Y_{j-1}\,,
 \ \ \ \
 W_{j}=Z_{j}\,X_{j-1}\,,
 \ \ \ \ \ldots\,,
 \ \ \ \ j = {K-1}, {K-2}, \ldots\,,
  \label{sillytoe}
 \ee
i.e., via recurrences
 \be
 Z_{j-1}=Z_j^{-1}\,
 Y_j\,,
 \ \ \ \
 Y_{j-1}=Z_j^{-1}\,
 X_j\,,
 \ \ \ \
 X_{j-1}=Z_j^{-1}\,
 W_j\,,
 \ \ \ \ \ldots\,,
 \ \ \ \ j = {K-1}, {K-2}, \ldots\,.
  \label{lytoe}
 \ee
The bra-vectors of Eq.~(\ref{dva})
can be then
given their ultimate, exclusively $Z_j-$dependent form,
 \be
 {\cal T}_{[j]}: |\psi\kt \ \to \ \br_{[j]} \psi| =
 \br \psi|\,Z_{K-1}\,Z_{{K-2}}\, \ldots\,Z_{j+1}\,,
 \ \ \ \ \ j={K-2}, {K-3},\ldots, 1,0\,.
  \label{altoe}
 \ee
All of these, by assumption, mutually non-parallel
bra-vectors have the same structure. The formula
reproduces the
${K}=3$ and ${K}=2$ special cases
appearing in the quasi-Hermitian quantum mechanics of
reviews~\cite{ali,Geyer,SIGMAdva}.
The ``non-metric''
operator $Z_{0}\ $ playing the role of
the physical SP Hamiltonian
does not enter our final definition (\ref{altoe})
of course.

We are now prepared to
define, in all of our Hilbert spaces ${\cal R}_j$, the
respective inner products $(
\psi_1,\psi_2)_{{\cal R}_{j}}\,\equiv\,
\br_{[j]} \psi_1|\psi_2\kt$
in recursive manner,
 \be
 \br_{[j]} \psi_1|\psi_2\kt
 =
 \br_{[j+1]}  \psi_1|Z_{j+1}|\psi_2\kt
  \,,\ \ \ \
 j={K}-2, {K}-3, \ldots,1,0\,.
 \label{reccu}
 \ee
Only the ``friendliest'', $j={K-1}$ item is allowed to be
written in the conventional Dirac's bra-ket form without subscripts again, $(
\psi_1,\psi_2)_{{\cal R}_{{K-1}}}=\br \psi_1|\psi_2\kt$.

\begin{table}[h]
\caption{List of auxiliary Hermiticity relations.}
 \label{o3x} \vspace{.4cm}
\centering
\begin{tabular}{||c|c|c|cccc||}
    \hline \hline
      &  & metric & \multicolumn{4}{|c|}
 {operator product}
\\
    $j$ & {\rm space }& $Z_j$ &$Z_jZ_{j-1}=Y_{j}$&$Z_jY_{j-1}=X_{j}$&
    $Z_jX_{j-1}=W_{j}$
    &\ldots
 \\
 \hline \hline
  & &&&&
  & \\
    $1$&
    ${\cal R}_1$
       &
  $Z_1=Z_1^{\ddagger(1)}$ &
  $Y_1=Y_1^{\ddagger(1)}$& & &
       \\
    $2$&
    ${\cal R}_2$
       &
  $Z_2=Z_2^{\ddagger(2)}$ &
  $Y_2=Y_2^{\ddagger(2)}$&
  $X_2=X_2^{\ddagger(2)}$ & &
       \\
    $3$&
    ${\cal R}_3$ &
  $Z_3=Z_3^{\ddagger(3)}$ &
  $Y_3=Y_3^{\ddagger(3)}$ &
  $X_3=X_3^{\ddagger(3)}$&
  $W_3=W_3^{\ddagger(3)}$&
             \\
      \vdots   &\vdots&\vdots&\vdots&\vdots&\vdots& $\ddots$\\
\hline \hline
\end{tabular}
\end{table}

\subsection{Systematic replacements of Hermiticities by pseudo-Hermiticities}

In
every inner-product space ${\cal R}_j$
we postulated the $j-$dependent Hermiticity property~(\ref{tri})
of $Z_j$ (cf. also the third column in Table \ref{o3x}).
At $j=0$ this relation
represents the most important
dynamical property of the Hamiltonian
while the rest of the list
concerns the $j-$dependent metrics.
From the point of view of the users working
in ${\cal H}_{math}$,
only the $j=N_{}$ item
 $$
 Z_{N_{}}
 =Z^{\ddagger (N_{})}_{N_{}}=Z^\dagger_{N_{}}
 $$
can directly be tested and verified.
The rest of the list with $j<N_{}$
will only be accessible to the verification
after its pull-down to  ${\cal H}_{math}$.

In
the first step of such a
process, in the manner proposed in \cite{PLA},
every Hermiticity relation
for a metric  (\ref{tri})
becomes re-interpreted
as a crypto-Hermiticity
requirement
 \be
 Z_j^{\ddagger(j+1)}\,Z_{j+1}=Z_{j+1}\,Z_j\,,\ \ \ \
 j=0,1,\ldots,K-2\,
 \label{koDYDY}
 \ee
imposed upon the same operator $Z_j$ in
the different inner-product context of
the neighboring,
more
friendly
space ${\cal R}_{j+1}$.
The series of transformations of the picture
to the more
friendly
space can be iterated.
For every
initial choice of the subscript $j$
the ultimate goal of the iterations
is to reach the formally equivalent representation
of property (\ref{tri}) in the mathematically optimal
space ${\cal H}_{math}={\cal R}_{K-1}$.


For illustration purposes let us now return to
the
${K}=3$ scenario
as discussed in \cite{PLA}.
The first-step quasi-Hermiticity
(\ref{koDYDY}) has
been shown there to imply
the
charge-pseudo-Hermiticity
of the
Hamiltonian
as well as the
parity-pseudo-Hermiticity
of the charge
(see
Table
Nr.~1
in {\it loc. cit.}).
The proof was based
on the antilexicographically ordered
relations
(\ref{sillytoe})
and on the observation that
$Z_1\,Z_0=Y_1=Y_1^{\ddagger(1)}$.
The latter relation is interesting also {\it per se\,}
because
the operator-product quantity
$Y_1$ is of an immediate phenomenological interest in the
random matrix theory \cite{joshua} or in the open-system physical
context (cf. Ref.~\cite{bian} where such a quantity
has been found conserved).

In the closed-system setting of Ref.~\cite{PLA}
the construction has been completed
by
the quasi-Hermiticity
$Y_1^{\ddagger(2)}\,Z_2=Z_2\,Y_1$
holding in
${\cal R}_2={\cal H}_{math}$.
This relation
proves equivalent
to the
parity-charge pseudo-Hermiticity
of the admissible Hamiltonians.
In this sense, the explicit reference to
the auxiliary space
${\cal R}_1$ and even to
the physical Hilbert space ${\cal R}_0$ itself becomes,
from the practical user's point of view,
redundant.

\begin{table}[h]
\caption{
Pseudo-Hermiticity properties of operators $Z_{k}$ at $k<j$.}
 \label{zo3x} \vspace{.4cm}
\centering
\begin{tabular}{||c|c|cccc||}
    \hline \hline
    $j$&&$k=j-1$&$k=j-2$&$k=j-3$&\ldots\\
    \hline \hline
     &  &
 {pseudo-metric}&
 {pseudo-metric}&
 {pseudo-metric}
 &
\\
    & {\rm space }
 &$Z_j
 $&$Y_{j}=Z_j
 Z_{j-1}
 $ &
    $X_{j}=Z_j
    Z_{j-1}Z_{j-2}
    $&\ldots\\
  \hline \hline
          &
    & & &    &
       \\
    $1$
       &
    ${\cal R}_1$&
   $Z_0^{\ddagger(1)}\,Z_1=Z_1\,Z_0$& & &
       \\
    $2$
       &
    ${\cal R}_2$  &
  $Z_1^{\ddagger(2)}\,Z_2=Z_2\,Z_1$&
  $Z_0^{\ddagger(2)}\,(Z_2\,Z_1)=(Z_2\,Z_1)\,Z_0$& &
       \\
    $3$
       &
    ${\cal R}_3$  &
  $Z_2^{\ddagger(3)}\,Z_3=Z_3\,Z_2$  &
  $Z_1^{\ddagger(3)}\,(Z_3\,Z_2)=(Z_3\,Z_2)\,Z_1$&
  $Z_0^{\ddagger(3)}\,(Z_3\,Z_2\,Z_1)=(Z_3\,Z_2\,Z_1)\,Z_0$&
       \\
      \vdots  &\vdots&\vdots&\vdots&\vdots & $\ddots$\\
\hline \hline
\end{tabular}
\end{table}

\subsubsection{$K>3$}

Once we move to $K>3$, every round of the formal
pseudo-Hermitization process can be perceived as a relegation
of Hermiticity to a more user-friendly space.
After a completion of the whole process
one obtains an $N_{}-$plet of the relevant relations
as sampled,
in Table \ref{zo3x},
at $j=N_{}$.
The step-by-step derivation of these relations is straightforward.
Indeed, the
second
round of the relegations involves
the
$({K-1})-$plet of operator products
 \be
 Z_j\,Z_{j-1}=Y_j=Y_j^{\ddagger(j)}\,,\ \ \ \
 j=1,2,\ldots,{K-1}\,.
 \label{Ylt}
 \ee
The  manifest Hermiticity of these products in
${\cal R}_{j}$
is
equivalent to their quasi-Hermiticity property
valid in the next space
${\cal R}_{j+1}$,
 \be
 Y_j^{\ddagger(j+1)}\,Z_{j+1}=Z_{j+1}\,Y_j
 \,,\ \ \ \
 j=1,2,\ldots,{K-2}\,.
 \label{Yrt}
 \ee
In full analogy with the special case of ISP, each operator
$Z_{j+1}\,$ plays here the role of an intermediate,
subscript-dependent
Hilbert- or Krein-space inner-product metric.

In the subsequent round of algebraic manipulations
we start from the the $Y_{j-1}-$containing operator products
 \be
 Z_j\,Y_{j-1}=X_j=X_j^{\ddagger(j)}\,,\ \ \ \
 j=2,3,\ldots,{K-1}\,.
 \ee
We replace their Hermiticity valid in
${\cal R}_{j}$ by the
quasi-Hermiticity
 \be
 X_j^{\ddagger(j+1)}\,Z_{j+1}=Z_{j+1}\,X_j
 \,,\ \ \ \
 j=2,3,\ldots,{K}-2\,
 \ee
postulated in the next Hilbert space
${\cal R}_{j+1}$. This is
in fact the key consequence of
our
assumption of existence of
the chain of Hermitian conjugations,
i.e., of several inner-product spaces.

The
core
of our message is the recommendation of
a systematic and exhaustive
iterative transfer of the description of all of the operators of
interest to ${\cal R}_{K-1}$. Such a process
is, therefore, naturally continued by the fourth round, with Hermiticities
 \be
 Z_j\,X_{j-1}=W_j=W_j^{\ddagger(j)}\,,\ \ \ \
 j=3,4,\ldots,{K-1}\,
 \ee
replaced by quasi-Hermiticites
 \be
 W_j^{\ddagger(j+1)}\,Z_{j+1}=Z_{j+1}\,W_j
 \,,\ \ \ \
 j=3,4,\ldots,{K-2}\,,
 \ee
etc. Tables~\ref{o3x} (implicitly) and \ref{zo3x}
(explicit when setting $j=N_{}$) offer a sample of the ultimate results
of these systematic replacements.

\section{Physics behind the multi-space mathematics\label{opere}}

The practical implementation
of the GSP formalism beyond its ISP special case
where we had ${N_{}}=K-1=2$
is a truly challenging task. In an indirect support
of its potential relevance let us
mention our older paper \cite{cptbe}.
At the time of its publication there was,
naturally, no
consistent SP
theory with $K=4$ at our disposal.
As long as the
model in question
exhibited a
nonlinear supersymmetry,
we needed
the charge ${\cal C}$
expressed as a product of two operators.
In the light of ansatz (\ref{fak}),
the physical Hilbert-space metric
would have to be triply factorized, therefore,
$\Theta_3=Z_3Z_2Z_1$.
In principle, the $j={N_{}}=3$ line of
our present Table \ref{zo3x}
would apply.
In retrospective, the non-availability of
a consistent GSP quantum theory with $K=4$
was precisely the reason
why the results of our study \cite{cptbe}
remained incomplete.

\subsection{Observability status of the charge at arbitrary $K$}

In the general$-K$ multiplet (\ref{multpl}) one
has to treat  the first item  ${\cal H}_{math}=
{\cal R}_{K-1}$
of the sequence as
the mathematically optimal, preferred and
manipulation-friendly Hilbert space.
The last item ${\cal H}_{phys}=
{\cal R}_{0}$
of the same sequence is, in contrast, the only Hilbert space
in which
the mean values of all of the observables
carry  the standard probabilistic interpretation, i.e., in which
the inner product
is
correct and physical.
All of the other inner-product spaces
play just an interpolative and auxiliary role.
In Table~\ref{zo3x}, in particular, one has to
select and, in practice,
work just with the single row where $j={N_{}}=K-1$.

At any $K$,
in the manner fully consistent with Stone theorem \cite{Stone},
Hamiltonians $H=Z_0\,$ have to be self-adjoint
in
${\cal H}_{phys}= {\cal R}_0$.
Simultaneously,
{\em all\,} of the relevant mathematical constructions
have to be performed
in working space
${\cal R}_{N_{}}={\cal H}_{math}$. Contextually,
the same Hamiltonian
can be called
Hermitian or non-Hermitian.
In the Hermitian case one simply
returns to the conventional probabilistic
interpretation of the unitary system in question.
In this sense
the physical interpretation of the predictions of the GSP
theory using $K>1$
remains unchanged.
Only our mathematical working space
ceases to be equivalent to the correct
physical Hilbert space
so that the
operators $\Lambda$ representing the observables
are defined, in
${\cal R}_{N_{}}={\cal H}_{math}$,  as non-Hermitian.

The physics behind the ``non-Hermitian''
theory acquires the universal, $K-$independent
CHSP form,
with its various
phenomenological aspects extensively discussed in Ref.~\cite{Geyer}.
In our present notation
one merely
represents
${\cal H}_{phys}={\cal R}_{0}$ in ${\cal R}_{N_{}}$
using a formal replacement
of the trivial metric in ${\cal R}_{N_{}}={\cal H}_{math}$
by
its
amended physical alternative.
This enables us to work, simultaneously, with
the two inner products
such that
 $$
 \br_{} \psi_a|\psi_b\kt=
 \br_{[N_{}]} \psi_a|\psi_b\kt\neq
 \br_{[0]} \psi_a|\psi_b\kt=
 \br_{[N_{}]} \psi_a|\Theta|\psi_b\kt=
 \br_{} \psi_a|\Theta|\psi_b\kt\,,
 $$
keeping in mind that only the latter inner product has
the standard probabilistic interpretation.

In the most elementary nontrivial scenario we
set $N_{}=2$ and
get the ISP formalism with $\Theta=Y_2$.
Equally well, we may
choose any larger $N_{}>2$ and
get the genuine GSP formalism.
Under both of these arrangements our Hamiltonian
only carries
its conventional physical
meaning of an observable in ${\cal R}_0={\cal H}_{phys}$.
Hence, even the predictions of
the generalized theory remain
probabilistic. The model-building process
does not proceed in the conventional physical Hilbert space
${\cal H}_{phys}$ but rather, step-by-step,
along
the sequence of
manifestly unphysical but perceivably user-friendlier
mathematical representation spaces ${\cal R}_{j}$
with the decreasing subscript $j$.

There exist two
main distinguishing features of the GSP models with $K>2$.
The first one
is formal: the correct physical metric
$\Theta$ which makes our Hamiltonian $H$ self-adjoint
is equal to an $N_{}-$term
operator product
which generalizes Eq.~(\ref{fak}).
The first few illustrative examples
are displayed as the $k=0$ items in Table \ref{zo3x}.

The second distinguishing feature of the GSP models is rather
serendipitious: One can notice that after a tentative
premultiplication by the charge ${\cal C}=Z_1$ from the right, also
the
$k=1$ pseudo-Hermiticities as sampled in Table \ref{zo3x}
acquire, unexpectedly, the form equivalent to
the precise crypto-Hermiticity  criterion
 \be
 {\cal C} \neq
 {\cal C}^\dagger=\Theta
 \,{\cal C}\,\Theta^{-1}\,.
 \ee
Once we compare this relation with Eq.~(\ref{quasi}) above
we see that the most unexpected but strictly physical second
characteristics of the GSP
theory is that it guarantees the observability of
the crypto-Hermitian Hamiltonian (i.e., energy)
{\em together with\,} the observability of
the crypto-Hermitian charge.

\subsection{Recurrences for conjugations}

The $j-$th-Hilbert-space Hermiticity [e.g., (\ref{tri})] and the
related $(j+1)-$th-Hilbert-space quasi-Hermiticity of an arbitrary
linear operator $\Lambda$ can be made explicit in recurrent manner,
 \be
 \Lambda^{\ddagger(j)} =(Z_{j+1})^{-1}\,
 \Lambda^{\ddagger(j+1)}\,Z_{j+1}\,,
 \ \ \ \ \ j=0,1,\ldots,{K-2}\,.
 \label{24}
 \ee
This pull-down of the conjugation connecting the neighboring Hilbert
spaces is an elementary consequence of definition (\ref{reccu})
which
played just a marginal role in Ref.~\cite{PLA} at $K=3$
(cf. equation Nr. 10 in
{\it loc. cit.}). The importance of relation (\ref{24})
grows with the growth of $K$.
At any $K\geq 3$
it leads to the
closed-form definition
 \be
 \Lambda^{\ddagger(j)} =
 \Theta_{(K-1,j)}^{-1}\,
 \Lambda^{\dagger}
 \,\Theta_{(K-1,j)}
 \,,\ \ \ \ \
  \Theta_{(K-1,j)}=Z_{{K-1}}\,Z_{{K-2}}\,\ldots\,Z_{j+2}\,Z_{j+1}\,
 \label{closef}
 \ee
of the $j-$th conjugation
in terms of the conventional one as defined in ${\cal H}_{math}$.

For illustration one could return to Tables \ref{o3x} and \ref{zo3x}.
A detailed inspection of the Tables
indicates that
at any preselected ${K}$
the relegation of the Hermiticity
will terminate only after we manage to exhaust
the whole set of products (\ref{sillytoe}).
The systematically iterated
step-by-step replacements
 $$
 Z_0 \, (=H) \to Y_1 \to
 X_2 \to  W_3 \to \ldots \to \Theta
 $$
will then guarantee
the self-adjointness of the Hamiltonian
in the correct space ${\cal H}_{[phys}={\cal R}_{{0}}$
re-expressed as its
quasi-Hermiticity
in ${\cal H}_{math}={\cal R}_{{K-1}}$.

An analogous chain of relegations of Hermiticity
remains applicable
to the charge and/or to any
other operator of phenomenological relevance.
In the spaces which are infinite-dimensional
it would be necessary to
discuss multiple technical questions
concerning the domains of operators, etc.
Here, we
skipped these questions and
restricted our attention just to the detailed
discussion of what could be called
the underlying algebraic relations and
symmetries.

\subsubsection{$K=3$}

At ${K}= 3$ the closed-form solution of the
GSP
recurrences
was given in~\cite{PLA}. For methodical
purposes it still makes sense to start the presentation of
this solution
by the tutorial re-derivation of the
first nontrivial ${K}=3$ ISP
pattern.
Indeed, in our present notation
the correspondence between  ${\cal H}_{phys}$ and
${\cal H}_{math}$
can be seen as mediated either by a single-step
simplification
of the inner product
$\br_{[0]} \psi_a|\psi_b\kt \
\to \ \br_{} \psi_a|\psi_b\kt$ (at $K=2$), or by a two-step
realization in which the preparatory step
$\br_{[0]} \psi_a|\psi_b\kt \
\to \ \br_{[1]} \psi_a|\psi_b\kt$ is followed by the final step
$\br_{[1]} \psi_a|\psi_b\kt \
\to \ \br_{} \psi_a|\psi_b\kt$ (at $K=3$).

Naturally,
after an elementary algebraic exercise
one can show that the phenomenological consequences of the
two recipes are equivalent.
Indeed, although the latter, ISP pattern with $K=3$ is prescribed
by the three inner-product-dependent Hermiticity
relations (\ref{tri}) which read
 \be
 Z_0=Z_0^{\ddagger(0)}\,,\ \ \ \
 Z_1=Z_1^{\ddagger(1)}\,,\ \ \ \
 \underline{ Z_2=Z_2^{\ddagger(2)}}\,,
 \label{2tri}
 \ee
only
the last, underlined item is
in its final form living in ${\cal H}_{math}={\cal R}_2$.
For the other two relations
we still have to find their representation
using the ${\ddagger(2)}-$marked
Hermitian conjugation defined
in the user-friendliest and
preferred
Hilbert space
${\cal R}_{2}={\cal H}_{math}$.
For this purpose let us recall Eq.~(\ref{koDYDY}) and
re-express the two former relations in the respective
quasi-Hermitian
forms,
 \be
 Z_0^{\ddagger(1)}\,Z_{1}=Z_{1}\,Z_0\,,\ \ \ \
 \underline{Z_1^{\ddagger(2)}\,Z_{2}=Z_{2}\,Z_1}\,.
 \label{2koDYDY}
 \ee
The second, underlined formula
is final as it already lives in ${\cal H}_{math}$.
The former item
requires a further translation using
the first formula in (\ref{Ylt}),
 \be
 Z_1\,Z_{0}=Y_1=Y_1^{\ddagger(1)}\,.
 \label{2Ylt}
 \ee
This is easily transferred to
${\cal H}_{math}$
via Eq.~(\ref{Yrt}),
 \be
 Y_1^{\ddagger(2)}\,Z_{2}=Z_{2}\,Y_1
 \,.
 \label{2Yrt}
 \ee
What remains to be done is the
insertion of definition (\ref{2Ylt}),
 \be
 Z_0^{\ddagger(2)}\,Z_1^{\ddagger(2)}\,Z_{2}=Z_{2}\,Z_1\,Z_{0}
 \,
 \label{i2Yrt}
 \ee
and the incorporation of Eq.~(\ref{2koDYDY}) yielding
our final underlined equation
 \be
 \underline{Z_0^{\ddagger(2)}\,\Theta_2=\Theta_2\,Z_{0}}
 \,,\ \ \ \ \ \ \Theta_2=Z_{2}\,Z_1\,.
 \label{i2Yrtb}
 \ee
This is the ultimate form of the quasi-Hermiticity of the
Hamiltonian written in terms of the Hilbert space metric
$\Theta_{K-1}$ at ${K}=3$.

We see that all of the auxiliary,
intermediate Hilbert spaces
${\cal R}_j$ with $0<j<{K-1}$ have been successfully eliminated.
In terms of the operator-product metric $\Theta_2$
the observability of $H$ is fully guaranteed, strictly
in the spirit of review \cite{Geyer},
by
its quasi-Hermiticity property (\ref{i2Yrtb}) in ${\cal H}_{math}$.
A consistent quantum model
can be constructed
in which all observables $\Lambda$ would
have to satisfy the
same quasi-Hermiticity relation as $H$ itself does.
Besides $Z_0$, as we already mentioned, also operator
$Z_1$ represents an observable
quantity [for proof it is sufficient to
pre-multiply the underlined equation in Eq.~(\ref{2koDYDY})
by $Z_1$ from the right]. Analogously, the product
$Z_2\,Z_1$ can represent another observable (the proof is similar)
while
$Z_2$, when standing alone, cannot.



\subsubsection{$K=4$}

At $N_{}=3$, the quadruplet of the Hermiticity relations (\ref{tri})
can be split in the
underlined final rule for $\underline{Z_3=Z_3^{\ddagger(3)}}$
in ${\cal H}_{math}$ and the remaining auxiliary triplet
 \be
 Z_0=Z_0^{\ddagger(0)}\,,\ \ \ \
 Z_1=Z_1^{\ddagger(1)}\,,\ \ \ \
 Z_2=Z_2^{\ddagger(2)}\,,
 \label{3tri}
 \ee
to be replaced by its quasi-Hermitian analogue (\ref{koDYDY}).
This yields, first of all, the last,
underlined relation $\underline{Z_2^{\ddagger(3)}\,Z_{3}=Z_{3}\,Z_2}\ $
for $Z_2$ which is already in its desired final form.
In ${\cal H}_{math}={\cal R}_{3}$ this makes
the quasi- or pseudo-Hermiticity of
$Z_2$ perceived
as mediated,
under appropriate mathematical conditions,
by metric $Z_{3}$.
The remaining doublet of relations
 \be
 Z_0^{\ddagger(1)}\,Z_{1}=Z_{1}\,Z_0\ (=Y_1)\,,\ \ \ \
 Z_1^{\ddagger(2)}\,Z_{2}=Z_{2}\,Z_1\ (=Y_2)
 \label{3koDYDY}
 \ee
is restricted by the respective Hermiticity requirements
$Y_1=Y_1^{\ddagger(1)}$ and $Y_2=Y_2^{\ddagger(2)}$
[cf. Eq.~(\ref{Ylt})].
Using relation (\ref{Yrt}) they may be both quasi-Hermitized,
 \be
 Y_1^{\ddagger(2)}\,Z_{2}=Z_{2}\,Y_1\ (=X_2)
 \,,\ \ \ \
 \underline{Y_2^{\ddagger(3)}\,Z_{3}=Z_{3}\,Y_2\ (=X_3)}
 \,.
 \label{3Yrt}
 \ee
The second, underlined item is already
written in the correct space ${\cal H}_{math}={\cal R}_{3}$.
The insertion of $Y_2$ from the second line of Eq.~(\ref{3koDYDY})
gives
 $
 \underline{Z_1^{\ddagger(3)}\,(Z_{3}\,Z_{2})=(Z_{3}\,Z_2)\,Z_{1}}
 $, i.e., the correct final rule for $Z_1$ formulated in ${\cal H}_{math}$.

We are left with the first item in (\ref{3Yrt}).
It can readily be quasi-Hermitized in ${\cal H}_{math}$,
 \be
 X_2^{\ddagger(3)}\,Z_{3}=Z_{3}\,X_2\,.
 \ee
After a series of elementary insertions we finally get
 \be
 \underline{Z_0^{\ddagger(3)}\,\Theta_3=\Theta_3\,Z_{0}}
 \,,\ \ \ \ \ \ \Theta_3=Z_{3}\,Z_{2}\,Z_1\,.
 \label{i3Yrt}
 \ee
These conclusions are again
summarized in Table~\ref{zo3x}.
The role of an observable which would be manifestly self-adjoint in
${\cal H}_{phys}$ can now be played by $Z_1$, by the product $Z_2\,Z_1$ and
by the product $Z_3\,Z_2\,Z_1$, but not by the product $Z_3Z_2$ or
by the operators $Z_2$ or $Z_3$ standing alone (the proofs are
analogous to the ones given above). Thus, in the
more conventional notation
one can conclude that  the ``hidden Hermiticity'' (i.e., the
observability) status holds for $H$ and for the charge ${\cal
C}=Z_1$. The observability status of the ``other charge'' ${\cal D}=Z_2$
(i.e., the validity of condition ${\cal
D}^\dagger\,\Theta_3=\Theta_3\,{\cal D}$) could only be achieved,
in a sufficiently elementary manner,
under an additional commutativity requirement ${\cal D}\,{\cal
C}={\cal C}\,{\cal D}$. Otherwise,
the observability status of ${\cal D}$ would require
a complicated analysis as sampled
in Ref.~\cite{arabky}.



\subsubsection{General $K$}

Among the ${K}$ Hermiticity relations (\ref{tri}) valid (or
postulated) for the operators $Z_j$, the last one (with $j=K-1$) is
a ${\cal H}_{math}-$space property of $Z_{K-1}$. Among the ${K-1}$
Hermiticity relations (\ref{koDYDY}) for the two-term products of
$Z_j$s, the last one (with $j=K-2$) is always a valid Hermiticity
property of $Y_{K-1}$ (i.e., a valid quasi- or pseudo-Hermiticity
property of $Z_{{K-2}}$) in the same ${\cal H}_{math}-$space. Etc.
Along this line we finally arrive at $j=0$ and Hamiltonian $H=Z_0$.
An elementary proof by mathematical induction really yields, in the
space ${\cal R}_{K-1}={\cal H}_{math}$, the quasi-Hermiticity
relation
 \be
 \underline{Z_0^{\ddagger({K-1})}\,\Theta_{K-1}\,\equiv\,
  Z_0^{\dagger}\,\Theta_{K-1}=\Theta_{K-1}\,Z_{0}}
 \,,\ \ \ \ \ \ \Theta_{K-1}=Z_{{K-1}}\,Z_{{K-2}}\,\ldots\,Z_{2}\,Z_1\,.
 \label{iYrt}
 \ee
This is our ultimate, experimentally relevant property of the
Hamiltonian of the system in question.

In the latter relation the ultimate metric operator $\Theta_{K-1}$
is factorized in a way which is compatible with the last $j=0$ item
in formula (\ref{altoe}) defining the ultimate and correct physical
bra-vector. An internal consistency of the GSP formalism is
confirmed. The observation also opens the possibility of a removal,
in the nearest future, of our present methodical limitation of
attention to the specific SP formulation of quantum mechanics.

In this direction, two next steps of possible methodical development
seem most promising. In the first one the present stationary GSP
theory could be replaced by its non-stationary generalization. This
would lead to the implementation of the present metric-factorization
idea in the non-SP context of the so called non-Hermitian
interaction picture (interested readers might find its $K=2$
description in \cite{NIP}).

In the second, alternative direction of research our present
factorization of the metric could also inspire a further progress in
the context of quantum statistical mechanics. In it, one would
merely replace a pure state characterized, in our present notation,
by the elementary projector
 \be
 \pi(t)=|\psi(t)\kt \,\frac{1}{\br
 \psi_{[0]}(t)|\psi(t)\kt}\,\br \psi_{[0]}(t)|\,
 \ee
by the non-Hermitian density matrix
 \be
 {\varrho}(t)=\sum_{k}|\psi^{(k)}(t)\kt
 \,\frac{p_k}{\br \psi^{(k)}_{[0]}(t)|\psi^{(k)}(t)\kt}\,\br
 \psi^{(k)}_{[0]}(t)|
 \,,
 \ \ \ \ \ \
 \sum_{k} p_k=1\,
 \ee
which would describe the probability distribution of a statistical
mixture of states. This density matrix would have the well known
physical meaning combining the probability-based experimental
preparation of a statistical quantum system with its generalized
$K>1$ theoretical description.

\section{Generalized Dyson maps\label{dymas}}

In the majority of applications of conventional quantum theory
the information about dynamics (i.e.,
about the realistic self-adjoint
Hamiltonian $\mathfrak{h}$) is
extracted, directly or indirectly, from the principle of
correspondence.
All interest in some unconventional,
crypto-Hermitian Hamiltonians emerges only
after one encounters some unsurmountable technical difficulties.
One then imagines that
the conventional formulations of quantum dynamics need not be optimal.
In such a situation, Dyson \cite{Dyson}
found, purely empirically, that his calculations prove perceivably
simplified after a reformulation of the SP theory in which one would
work with two inner-product spaces (i.e., in our
present notation, with ${K}=2$).

The motivation of the
Dyson's
construction was pragmatic,
based on a good intuitive guess of the
operator $\Omega$ representing correlations.
In fact, the strong dependence of the success
on this guess was one of the main
weak points of the strategy.
Thus, one has to ask how should one
amend this aspect, and how could one make the choice of
the correlators $\Omega$
more robust and flexible.

\subsection{Multi-step preconditionings\label{gedymas}}

In the most elementary $K=2$ CHSP scenario
it is not easy to keep the physics given by
the respective candidates $\mathfrak{h}$ and $\Omega\,$
for the Hamiltonian and mapping
fully under control.
One must be even more careful when
using the inverted flowchart $H \to \Theta \to \Omega$
as discussed by Scholtz et al \cite{Geyer}.
In an ideal situation,
the determination of an optimal $\Omega$
should not be based on the mere
educated guess.
One may expect an improvement of the results, in particular, after
the above-described multi-step
relegation of the
Hermiticity of
operators from ${\cal R}_{j}$ to ${\cal R}_{j+1}$.

What seems lost is
the reference
to the idea of
preconditioning, i.e., an intuition-based insight
in the dynamics. This is a shortcoming of the theory which damages the
appeal and popularity of the CHSP formalism in applications.
A step towards recovery came with PTQM, i.e., with the
proposal of
factorization (\ref{fak}).
Still,
a certain theoretical weakness of the recipe survived,
lying in a
somewhat mysterious status of the charge \cite{ali} as well as
in a manifestly non-Dysonian nature of the Bender's
factorization $\Theta_2={\cal PC}$ of the metric.

The gap has partially been filled in \cite{PLA}.
We found there that the charge
${\cal C}$ becomes in fact a very natural component of the
formalism.
Still, our
satisfaction remained limited because the operator-product
metric $\Theta={\cal PC}$ did not seem to be
easily re-factorized
(i.e., in effect, interpreted) in terms of a
single Dyson map $\Omega$.
Also
the most common re-factorization
$\Omega=\sqrt{\Theta}$ of the metric (as recommended, e.g.,
in \cite{ali}) may often prove too artificial.

In \cite{PLA},
precisely the latter artificiality impression
caused our loss of
interest in the Hermiticity-to-quasi-Hermiticity
relegations at ${K}\geq 4$.
We did not imagine that a consistent
concept of the Dyson maps
and of their compositions
might
very naturally be connected with
the notion of the space-dependent Hermiticity.

\subsection{Composition laws}

In
a way inspired by Eq.~(\ref{amenes}) let us
postulate a
factorization
 \be
 Z_j=\Omega_j^{\ddagger(j)}\,\Omega_j\,,\ \ \ \
 j=1,\ldots,{K-1}\,.
 \label{btri}
 \ee
Such an ansatz is, first of all, compatible with our present
$j-$dependent
self-adjointness (\ref{tri}) of $Z_j$ in ${\cal R}_j$.
It allows us to
reclassify the new family of the invertible operators
$\Omega_j$ as an upgraded
model-building-information input.

\begin{lemma}
\label{lemko1}
At ${K}=3$ let us assume that both
$Z_2$ and $Z_1$ are positive definite and
factorized via Eq.~(\ref{btri}). In ${\cal H}_{math}$,
the physical Hilbert-space metric $\Theta_2=Z_2Z_1$
is then factorized as follows,
 \be
 \Theta_2=\Omega_{21}^{\ddagger(2)}\,\Omega_{21}
 =\Omega_{21}^{\dagger}\,\Omega_{21}\,,
 \ \ \ \ \ \Omega_{21}=\Omega_{2}\,\Omega_{1}\,.
 \label{amen2}
 \ee
\end{lemma}
\begin{proof}
Relation (\ref{24}) implies that in the factorization postulate
$Z_1=\Omega_{1}^{\ddagger(1)}\,\Omega_{1}$ we can relegate the
Hermiticity of
$\Omega_{1}^{\ddagger(1)}=Z_2^{-1}\,\Omega_{1}^{\ddagger(2)}\,Z_2$.
The insertion in $\Theta_2=Z_2\,Z_1$ yields
the result.
\end{proof}

\begin{lemma}
\label{lemko2}
At ${K}=4$ let us assume that operators
$Z_j$ with $j=1,2$ and $3$ are positive definite and factorized via
Eq.~(\ref{btri}). In ${\cal H}_{math}$,
the physical Hilbert-space metric
is then factorized as follows,
 \be
 \Theta_3=Z_3Z_2Z_1
 =\Omega_{321}^{\ddagger(3)}\,\Omega_{321}
 =\Omega_{321}^{\dagger}\,\Omega_{321}\,,
 \ \ \ \ \ \Omega_{321}=\Omega_{3}\,\Omega_{2}\,\Omega_{1}\,.
 \label{amen3}
 \ee
\end{lemma}
\begin{proof}
As long as the superscript $^{\ddagger(j)}$ denotes
the usual Hermitian conjugation in
our unique representation space ${\cal H}_{math}$ at $j={K-1}=3$,
we have to eliminate, from
the definition of $\Theta_3=Z_3\,Z_2\,Z_1$, just the
two unusual Hermitian conjugation
superscripts $^{\ddagger(j)}$
with $j=1$ and $j=2$.
We need, for this purpose,
just the two items in
relations (\ref{24}), viz.,
\be
 Z_{2}\,\Lambda^{\ddagger(1)} =
 \Lambda^{\ddagger(2)}\,Z_{2}
 \label{3124}
 \ee
and
 \be
  Z_{3}\,\Lambda^{\ddagger(2)} =
 \Lambda^{\ddagger(3)}\,Z_{3}
 = \Lambda^{\dagger}\,Z_{3}\,,
 \label{3224}
 \ee
respectively. With $\Lambda=\Omega_1$
we get $$\Theta_3=Z_3\,\Omega_1^{\ddagger(2)}\,Z_2\,\Omega_1$$
from the former identity, while
we further get $$\Theta_3=\Omega_1^{\dagger}\,Z_3\,Z_2\,\Omega_1$$
from the latter identity.
In the third step, the factorization (\ref{btri}) at $j=2$
and the use of Eq.~(\ref{3224}) with $\Lambda=\Omega_2$
lead immediately to the final result.
\end{proof}

The structure of the latter proof reveals its recursive nature.
Along the same lines it is now easy to prove the general result
(the task is left to the readers).

\begin{thm}
\label{theodor}
At any integer ${K}\geq 2$
the existence of decompositions~(\ref{btri})
of the positive definite
operators $Z_j$ at $j=1,2,\ldots,{K-1}$
is formally equivalent to
the existence of a refactorization
 \be
 \Theta_{K-1}=
 \Omega_{{K-1}\ldots 321}^{\dagger}\,\Omega_{{K-1}\ldots 321}\,,
 \ \ \ \ \ \Omega_{{K-1}\ldots 321}=\Omega_{{K-1}}
 \,\ldots \,\Omega_{3}\,\Omega_{2}\,\Omega_{1}\,,
 \label{acto}
 \ee
of the Hilbert-space metric
 \be
 \Theta_{K-1}=Z_{{K-1}}\,Z_{{K-2}}\,\ldots\,Z_{2}\,Z_1\,
 \label{amenK}
 \ee
which is responsible for the
quasi-Hermiticity (\ref{iYrt})
of the Hamiltonian
$H=Z_0$ in ${\cal H}_{math}$.
\end{thm}

It makes good sense to emphasize
that the latter result is purely algebraic,
i.e., guaranteed as
valid in the finite-dimensional Hilbert spaces.
For the infinite-dimensional Hilbert spaces
the statement had to be complemented
by several necessary technical assumptions
like, typically, those concerning the domains and ranges of
operators, etc.
In this direction of the further development of the theory,
the recommended preparatory reading
may be found, e.g., in the recent
dedicated monograph \cite{book}.

\section{Conclusions\label{sumys}}

In the manner proposed and verified by Dyson \cite{Dyson},
a decisive simplification of difficult
calculations of certain properties
of complicated quantum systems
can sometimes be achieved via introduction of an
additional, auxiliary inner product \cite{PLA}.
Our present paper offers, in
essence, just a nontrivial $K-$space
technical extension and amendment of the
Dyson-inspired CHSP quantum theory.

We felt motivated by the observation
that the fairly complicated CHSP-related mathematics
weakened the appeal and
visibility of the innovative methodical aspects of the approach
\cite{MZbook}.
For this reason,
the current return of attention to the physics with
non-Hermiticites \cite{Christodoulidesb} was obviously
motivated by the
simplification of some technicalities
provided by the most popular PTQM philosophy \cite{Carl,ali}.

Applications of the updated PTQM
theory profited, first of all, from
the factorization
of metric.
Simultaneously, these applications suffered
from the loss of a direct contact
with the principle of correspondence
as offered, initially, by Eq.~(\ref{prvak}).
In our recent brief note \cite{PLA} we
addressed the problem.
We pointed out
that the somewhat mysterious PTQM-mediated charge-based
relegation of the Hermiticity
from the physical Hilbert space ${\cal H}_{phys}$ to
auxiliary ${\cal H}_{math}$
may be given a very natural conceptual background
when one decides to treat,
on equal footing,
the observability of the Hamiltonian $H$
and the observability of the charge ${\cal C}$.

More recently we revealed that the strict restriction of
the resulting ISP formalism
to the $K=3$ kinematical scenario given by Eq.~(\ref{tripl})
was in fact based on an elementary misunderstanding.
We felt discouraged by the
apparently counterintuitive
observation that at $K=3$ the parity ${\cal P}=Z_2$
seems to have lost its physical observability status.
Only later we imagined that such a loss is in
fact mathematically very natural, especially when one
builds the models without some suitable {\it ad hoc\,}
additional constraints (see, e.g., an
extensive discussion of this point in \cite{arabky}).

With the latter mental barrier removed, we still felt
discouraged by another, more technical observation
that whenever one tries to move to any nontrivial $K>2$
(i.e., to the kinematics using the larger inner-product-space
multiplets (\ref{multpl})), the properly generalized Dyson
maps seem to ``cease to behave nicely''.
After a more detailed analysis of the latter problem
(leading, finally, to the results sampled here
by Theorem \ref{theodor}), fortunately,
also the latter doubts have been eliminated.

We can conclude
that
the present systematic $K>3$ GSP extension
of the scope of the most economical and popular
SP formulation of quantum mechanics
will find efficient implementations
in all of the contexts in which one
finds the reasons for the replacement of the
Dyson-mapping-based
metric (\ref{amenes})
by its generalizations, be it
the
metrics (\ref{amen2}) or (\ref{amen3})
or the entirely general metric-factorization formulae
(\ref{acto}) or (\ref{amenK})
in which the number of the auxiliary generalized-charge factors
can be an arbitrary positive integer $N_{}=K-1$.

\newpage

\subsection*{Acknowledgements}

The author is grateful to Excellence project P\v{r}F UHK
2211/2022-2023 for the financial support.

\subsection*{Data Availability Statement}

Data sharing is not applicable to this article as no new data were
created or analyzed in this study.

\subsection*{Conflicts of Interest}

The author declares no conflict of interest.

\end{document}